\documentclass[a4paper,centertags,reqno,11pt]{amsart}

\usepackage{setspace}
\usepackage[headings]{fullpage}					
\usepackage[inline]{enumitem}   
\usepackage{array}												
\usepackage{amsmath}	        	
\usepackage{amsthm}							
\usepackage{amssymb}										
\usepackage{accents}						
\usepackage{extarrows}          
\usepackage{xfrac}              
\usepackage{url}						
\usepackage[round,comma,authoryear,sort&compress]{natbib} 
\usepackage{lmodern}
\usepackage{anyfontsize}
\usepackage[colorlinks,citecolor=blue,urlcolor=magenta,unicode,verbose,breaklinks=true]{hyperref}
\usepackage{breakcites}
\usepackage{rotating}
\usepackage{amscd}

\DeclareMathAlphabet\scr{U}{scr}{m}{n}
\DeclareMathAlphabet\scr{U}{scr}{m}{n}
\SetMathAlphabet\scr{bold}{U}{scr}{b}{n}
\DeclareFontFamily{U}{scr}{\skewchar\font'177}
\DeclareFontShape{U}{scr}{m}{n}{<-6>rsfs5<6-8>rsfs7<8->rsfs10}{}
\DeclareFontShape{U}{scr}{b}{n}{<-6>rsfs5<6-8>rsfs7<8->rsfs10}{}
\newtheorem{theorem}{Theorem}[section]

\newtheorem{corollary}[theorem]{Corollary}

\newtheorem{definition}[theorem]{Definition}

\newtheorem{lemma}[theorem]{Lemma}

\newtheorem{proposition}[theorem]{Proposition}
\newtheorem{remark}[theorem]{Remark}

\numberwithin{equation}{section}
\theoremstyle{plain}
\newtheorem{assumption}[theorem]{Assumption}

\def\softl{l\kern-0.35ex\raise0.1ex\hbox{'}\kern-0.15ex}

\newcommand{\dom}{\operatorname{dom}}
\newcommand{\sq}{\mathbin{\square}}
\newcommand{\R}{\mathbb{R}}
\newcommand{\E}{\mathrm{E}}
\newcommand{\Var}{\mathrm{Var}}
\newcommand{\m}{\mathrm{m}}
\newcommand{\MV}{\textsc{MV}}
\newcommand{\MMV}{\textsc{MMV}}

\newcommand{\SR}{\mathrm{SR}}
\renewcommand{\P}{\textsf{\upshape P}}
\newcommand{\Qu}{\textsf{\upshape Q}}
\renewcommand{\d}{\mathrm{d}}
\newcommand{\indicator}[1]{\mathbf{1}_{#1}}

\usepackage{letltxmacro}
\makeatletter
\let\oldr@@t\r@@t
\def\r@@t#1#2{%
\setbox0=\hbox{$\oldr@@t#1{#2\,}$}\dimen0=\ht0
\advance\dimen0-0.2\ht0
\setbox2=\hbox{\vrule height\ht0 depth -\dimen0}%
{\box0\lower0.4pt\box2}}
\LetLtxMacro{\oldsqrt}{\sqrt}
\renewcommand*{\sqrt}[2][\ ]{\oldsqrt[#1]{#2}}
\makeatother
\begin{document}

\title[Semimartingale Monotone Mean--Variance Portfolio Theory]{Semimartingale Theory of Monotone\\
Mean--Variance Portfolio Allocation}

\author{Ale\v{s} \v{C}ern\'{y}}

\address{Ale\v{s} \v{C}ern\'{y}\\
  Cass Business School\\
	City, University of London}

\email{ales.cerny.1@city.ac.uk}

\thanks{I would like to thank Fabio Maccheroni and Sara Biagini for many helpful discussions. I am also grateful to two anonymous referees for their comments.}

\subjclass[2010]{Primary: 05C38, 15A15; Secondary: 05A15, 15A18}
\keywords{monotone mean--variance; monotone Sharpe ratio; free cash-flow stream}
\dedicatory{Cass Business School, City, University of London}

\begin{abstract}
We study dynamic optimal portfolio allocation for monotone mean--variance preferences
in a general semimartingale model. Armed with new results in this area we revisit the work of \citet{cui.al.12} 
and fully characterize the circumstances under which one can set aside a non-negative cash flow while simultaneously improving the mean--variance efficiency of the left-over wealth. The paper analyzes, for the first time, the monotone hull of the Sharpe ratio
and highlights its relevance to the problem at hand.
\end{abstract}

\maketitle

\section{Introduction}

In a recent work \citet{cui.al.12} study an interesting situation where one can enhance the performance of a dynamically mean--variance efficient portfolio by setting aside a non-negative cash flow without lowering the Sharpe ratio of the remaining wealth distribution. Noting that mean--variance (MV) preferences are not time-consistent, \citet{cui.al.12} devise a new concept called \emph{time consistency in efficiency} which allows one to judge whether one can extract a  `free cash-flow stream' (FCFS) without affecting the efficiency of the mean--variance portfolio allocation. Their analysis is performed mostly in discrete time with square-integrable price processes whose returns are independent, but the authors also note that FCFS extraction is not possible in a continuous-time lognormal diffusion model. In subsequent work \citet{bauerle.grether.15} indicate that the FCFS extraction is not possible in any (suitably defined) complete market. \citet{trybula.zawisza.19} reach an identical conclusion in a specific (incomplete) diffusion setting. 

In this paper we approach the same subject along more classical lines to deepen the foregoing analysis both mathematically and conceptually. First, we consider a general semimartingale model  with only a mild $\sigma$--local square integrability condition on the price process. 
Second, we note that the extraction of FCFS for MV preferences can be fully understood by studying a simpler time-consistent expected utility maximization, which shows that any link between the existence of FCFS and time inconsistency of MV preferences is accidental.

Our strategy, in the first instance, is to link the existence of FCFS to portfolio maximization with monotone mean--variance preferences (MMV). In the second step we exploit a connection between the MMV preferences and the truncated quadratic utility whose individual ingredients have appeared in the work of \citet{cerny.03}, \citet{bental.teboulle.07}, \citet{filipovic.kupper.07}, \citet{maccheroni.al.09}, 
 and  \citet{cerny.al.12}. We provide a novel and systematic treatment of this connection which is of independent interest. As a by-product we then obtain an extension of the monotone mean--variance optimal portfolio analysis of \citet{maccheroni.al.09} to semimartingale trading.

In the general semimartingale setting outlined above we prove that it is possible to extract an FCFS while maintaining MV efficiency if and only if one can extract an FCFS and \emph{strictly improve} MV efficiency (Theorem~\ref{T:3}). Our work explicitly characterizes the upper limit of MV efficiency gain in terms of the monotone hull of the Sharpe ratio (SR). In Proposition~\ref{P:SRm} the monotone SR is shown to coincide, on an appropriate set, with the \emph{arbitrage-adjusted Sharpe ratio} of \citet{cerny.03}.

\section{Mathematical setup}

\subsection{Monotone mean--variance preferences}
Fix a time horizon $T>0$. We shall work on a filtered probability space $(\Omega,(\scr{F}_t)_{t\in [0,T]},\P)$ with $\scr{F}_0$ trivial. We write $L^p$ for $L^p(\Omega,\mathcal{F}_T,\P)$ with $p\in[0,\infty]$ and $L^p_+$ for the set of non-negative random variables in $L^p$. All probabilistic statements hold `$\P$--almost surely'.

Let $U:\R\to\R$ be the normalized quadratic utility
$$U(x)=x-x^2/2.$$
Define the expected utility functional $F:L^{0}\rightarrow \mathbb{[-\infty },\infty )$ by 
\begin{equation*}
F(X)=\E[U(X)]. 
\end{equation*}
Observe that $F$ is a proper and concave function on $L^{0}$. The \emph{effective domain} of a concave function $f$ on $L^{0}$ is defined in the standard way as
$$\dom f = \{X\in L^{0}\mid f(X)>-\infty \}.$$ 
In particular, we obtain $\dom F=L^{2}$.

Next, denote by $F_{\m}$ and $F_{\MV}$ the monotone and the cash-invariant hull of $F$, respectively, cf. 
\citet[Section 4]{filipovic.kupper.07},
\begin{align}
F_{\m}(X) &= \sup_{ Y\in L^0_+} F(X-Y), \label{eq:181213.1}\\
F_{\MV}(X) &= \sup_{c\in \R} \left\{ F(X-c)+c \right\}.    \label{eq:181213.2}
\end{align}
The easy proof of the next lemma is omitted.
\begin{lemma}
Functionals $F_{\m}$, $F_{\MV}$  are concave and proper on $L^0$ with
\begin{align*}
\dom F_{\m} &=L^0_+ - L^2_+,\\
\dom F_{\MV} &= L^2.
\end{align*}
Moreover, on their effective domains $F_m$ and $F_{\MV}$ obey the identities
\begin{align}
F_{\m}(X) &=\E[X\wedge 1] - \E[(X\wedge 1)^2]/2,\label{eq:Fm}\\
F_{\MV}(X) &= \E[X] -\Var(X)/2.\label{eq:Fc}
\end{align}
\end{lemma}

Finally, denote by $F_{\MMV}$ the monotone hull of the mean--variance preference,
\begin{equation}\label{eq:Fcm}
F_{\MMV}(X) =\sup_{Y\in L_{+}^{0}} F_{\MV}(X-Y).
\end{equation}
Observe that $F_{\MMV}$ restricted to $L^2$ is precisely the monotone mean--variance preference of \citet{maccheroni.al.09}. In our seting the effective domain of the monotonization is naturally somewhat larger,
$$ \dom F_{\MMV} = \dom F_{\m} = L^0_+ - L^2_+.$$

\subsection{Price processes and admissible strategies}
We assume there are $d\in\mathbb{N}$ risky assets and a risk-free bond with constant value $1$. For more details concerning the next assumption see \citet[Section~2.4]{biagini.cerny.11}.
\begin{assumption}\label{A:S}
The prices of risky assets are modelled by an $\R^d$--valued $\sigma$--locally square-integrable semimartingale $S$.
\end{assumption} 
Recall the definition of an absolutely continuous \emph{signed} $\sigma$--martingale measure for $S$ in \citet[Definition~2.3]{cerny.kallsen.07}. Denote the totality of such signed measures $\mathcal{M}^s$ and the subsets containing only absolutely continuous (resp. equivalent)  probability measures by $\mathcal{M}^a$ (resp. $\mathcal{M}^e$). Finally, for $l\in \{s,a,e\}$ define 
$$ \mathcal{M}^{l}_2 = \{\Qu\in \mathcal{M}^{l}\mid \d\Qu/\d\P \in L^2\}.$$
\begin{definition}\label{D:1}
We say that $\vartheta\in L(S,\P)$ is a tame strategy, writing $\vartheta\in \scr{T}$, if
$$\sup_{t\in[0,T]}\lvert\vartheta\cdot S_t\rvert \in L^2.$$
We say that $\vartheta\in L(S,\P)$ is an admissible strategy
\begin{itemize}
\item[---]  for the preference $F$, writing $\vartheta\in\scr{A}$, 
if $\vartheta\cdot S$ is a $\Qu$--martingale for every $\Qu\in\mathcal{M}_2^s$;
\item[---] for the preference $F_{\m}$, writing $\vartheta\in\scr{A}_{\m}$, 
if $\vartheta\cdot S$ is a $\Qu$--supermartingale for every $\Qu\in\mathcal{M}_2^a$.
\end{itemize}
\end{definition}
\noindent In this context we remark that the notion of admissibility in \citet{bauerle.grether.15}
is unsatisfactory because it does not rule out doubling strategies, and
therefore arbitrage, in continuous-time models. In particular, in their setting the Black--Scholes model is not arbitrage-free, see \citet[Section~6]{harrison.kreps.79}.

We will work under the following no-arbitrage assumption, see also \citet[Assumption~2.1]{cerny.kallsen.07}.
\begin{assumption}\label{A:M}
$\mathcal{M}_2^e$ is not empty.
\end{assumption}
\begin{theorem}\label{T:1}
Assume~\ref{A:S} and \ref{A:M}. For every $x\in\R$ one then has 
\begin{alignat}{3}
u(x) &= \sup_{\vartheta\in\scr{T}} F(x+\vartheta\cdot S_T) &&= \max_{\vartheta\in\scr{A}}F(x+\vartheta\cdot S_T) &&= 
\frac{1}{2}
-\frac{1}{2}\frac{\left(1-x\right)^2}{\min\limits_{\Qu\in\mathcal{M}^s_2}\E\left[\left(\frac{\d\Qu}{\d\P}\right)^2\right]},\label{ux}\\
u_{\m}(x) &=\sup_{\vartheta\in\scr{T}} F_{\m}(x+\vartheta\cdot S_T) &&= \max_{\vartheta\in\scr{A}_{\m}}F_{\m}(x+\vartheta\cdot S_T) &&=  
\frac{1}{2}-\frac{1}{2}\frac{\left((1-x)^+\right)^2}{\min\limits_{\Qu\in\mathcal{M}^a_2}\E\left[\left(\frac{\d\Qu}{\d\P}\right)^2\right]}.\nonumber 
\end{alignat}
\end{theorem}
\begin{proof}
The first statement follows from \citet[Lemma~2.4]{cerny.kallsen.07}. The second statement follows from Theorem~2.1 and Propositions~3.5 and 5.3 in \citet{biagini.cerny.20} specialized to $L^{\hat{U}}\sim L^2$ with the utility function 
$$x\mapsto x\wedge 1-(x\wedge 1)^2/2.$$ 
Although Definition~\ref{D:1} is a little narrower than the definition of tame strategies in \citet[Definition~5.1]{biagini.cerny.20}, the set of separating measures with density in $L^2$ remains the same, namely $\mathcal{M}_2^a$, and all arguments in \citet{biagini.cerny.20} go through. See also \citet[Proposition~6.4]{biagini.cerny.11}.
\end{proof}

\section{New characterization of monotone mean--variance preferences}

Define a concave  `cash indicator function' $C:L^{0}\rightarrow [-\infty ,\infty )$, 
\begin{equation*}
C(X)=\left\{ 
\begin{array}{ccl}
c&&\text{for }X=c,\ c\in \mathbb{R}; \\ 
-\infty&& \text{otherwise}.
\end{array}
\right. 
\end{equation*}
Let $D:L^{0}\rightarrow [-\infty ,\infty )$ denote the concave indicator function of the positive cone $L^0_+$,
\begin{equation*}
D(X)=\left\{ 
\begin{array}{ccl}
0&&\text{for }X\in L^0_+ ;\\ 
-\infty&& \text{otherwise.}
\end{array}%
\right. 
\end{equation*}

Let $f$, $g$ be two concave,  proper functions on $L^{0}$. With \citet{rockafellar.70} we define the supremal convolution of $f$ and $g$
\begin{equation*}
\left( f\sq g\right) (Z)=\sup \left\{ f(X)+g(Y)\mid X+Y=Z\right\}.
\end{equation*}
One easily verifies that \eqref{eq:181213.1} means $F_{\m} = F\sq D$ and \eqref{eq:181213.2} means $F_{\MV} = F\sq C$.

The key mathematical observation is that supremal convolution is a commutative and associative operation, so that we obtain, with no additional effort,
\begin{equation}\label{eq: CD}
\begin{CD}
F @>\sq\,D>> F_{\m}\\
@VV \sq\,C V @VV\sq\,C V\\
F_{\MV}@>\sq\,D>>F_{\MMV}
\end{CD}\qquad.\medskip
\end{equation}
Let us summarize the lessons from the commutative diagram \eqref{eq: CD}.
\begin{enumerate}
\item $F$ is the expected quadratic utility with $U(x)=x-x^{2}/2$ ;

\item $F_{\m}$  is the expected truncated quadratic utility, see \eqref{eq:Fm}, with 
\begin{align}\label{eq:Um}
U_{\m}(x) =x\wedge 1-\frac{\left( x\wedge 1\right) ^{2}}{2}=\frac{1-\left( \left( x-1\right) \wedge 0\right) ^{2}}{2};
\end{align}

\item $F_{\MV}$  is the mean--variance preference, see \eqref{eq:Fc};

\item The monotone mean--variance preference is classically computed by starting in the top left corner of the diagram \eqref{eq: CD} and then moving anti-clockwise: down and to the right. The resulting formula $F_{\MMV}=F_{\MV}\sq D$ corresponds to equation \eqref{eq:Fcm}. If instead one proceeds from the top left of diagram \eqref{eq: CD} by going clockwise, i.e., first to the right and then down, one obtains a seemingly different but equivalent expression $F_{\MMV}=F_{\m}\sq C$. Explicitly, this new formula reads 
\begin{equation}\label{eq:Fmc}
F_{\MMV}(X) =\sup_{c\in \R} \left\{\E\left[U_{\m}(X-c) \right]+c\right\},\qquad X\in L^0_+ - L^2_+ .
\end{equation}
\end{enumerate}

Let us first address the economic significance of formula \eqref{eq:Fmc}. It shows that maximization of monotone mean--variance preferences is essentially just maximization of time-consistent expected utility where at the outset one pre-commits to the correct level of $c$. This level is given by formula \eqref{eq:cm}, hence it is always non-negative and obtainable from the same expected utility maximization with $c=0$. Non-negative $c$ has the effect of increasing investor's local risk-aversion compared to $c=0$.

The existing literature characterizes monotone mean--variance preferences mostly by the formula \eqref{eq:Fcm} or its variational counterpart (both restricted to $L^2$, see \citealp{maccheroni.al.09}, equations 2.3--2.4)
\begin{equation*}
\left. F_{\MMV}\right\vert_{L^2}(X)=\inf \left\{ \E[ZX]+\Var(Z)/2\mid Z\in L_{+}^{2},\E[Z]=1\right\}.
\end{equation*}
In a separate strand, \citet[Proposition~2.1 and Theorem~4.2]{bental.teboulle.07} notice a link between variational preferences and cash-invariant hull of expected utility, which in their work is called the \emph{optimized certainty equivalent}. \citet[Theorem~7]{cerny.al.12} use this link to prove the formula \eqref{eq:Fmc} restricted to $L^\infty$. 

The proof of the equivalence between \eqref{eq:Fcm} and \eqref{eq:Fmc} by means of supremal convolution appears to be new. It is more direct that  the alternatives suggested in the literature and offers the additional advantage of working on the wider domain $L^0$.
\section{Monotone hull of the Sharpe ratio}
Let $\SR : L^0\to [-\infty,\infty]$ be the map
\begin{equation*}
\SR(X) = \left\{ 
\begin{array}{cl}
\E[X]/\sqrt{\Var(X)} & \text{for }X\in L^2 \\ 
-\infty & \text{otherwise} 
\end{array}\right.,
\end{equation*}
with the convention $1/0=\infty$, $-1/0=-\infty$, and $0/0=0$. 

Next, define $\SR_{\m} : L^0_+ - L_{+}^2 \to (-\infty,\infty]$ as the monotone hull of $\SR$, that is
$$\SR_{\m} = \sup_{Y\in L_{+}^{0}} \SR(X-Y). $$
\begin{proposition}\label{P:SRm}
Assume $X\in L^0_+ - L_{+}^2$ is such that $X^-\neq 0$ and 
\begin{equation}\label{eq:EX>0}
\lim_{K\to \infty}\E[X\wedge K]\in (0,\infty].
\end{equation} 
Then 
$$\sup_{\alpha\geq 0} F_{\m}(\alpha X)  $$
has a unique optimizer $\hat{\alpha}>0$ obtained as the unique solution of 
\begin{equation}\label{eq:FmFOC}
\E[X \indicator{\hat\alpha X\leq 1}] = \hat\alpha\E[X^2 \indicator{\hat\alpha X\leq 1}].
\end{equation}
Furthermore, 
\begin{equation}\label{eq:SRm2}
\SR_{\m}(X) = \SR((\hat\alpha X)\wedge 1)=\SR(X\wedge \hat{\alpha}^{-1})=\max_{K>0}\SR(X\wedge K).
\end{equation}
\end{proposition}
\begin{proof}
Define $f:\R_+\to \R$, $f(\alpha) = F_{\m}(\alpha X)$. Under our integrability assumptions on $X$ \eqref{eq:Fm} and dominated convergence yield
\begin{align*}
f'(\alpha)&=\E[X\indicator{\alpha X\leq 1}]-\E[\alpha X^2 \indicator{\alpha X\leq 1}],\\
f''_+(\alpha) &= -\E[X^2 \indicator{\alpha X< 1}],\\
f''_-(\alpha) &= -\E[X^2 \indicator{\alpha X\leq 1}].
\end{align*}
The derivative $f'$ is strictly decreasing on $\R_+$ with $f'>0$ near zero and $f'<0$ near infinity. As $f'$ is continuous on $(0,\infty)$ it has a unique root $f'(\hat\alpha)=0$ and by standard arguments this root is the global maximum of $f$ on $\R_+$ which proves 
\eqref{eq:FmFOC}. At the optimum the value function reads
\begin{equation}\label{eq:um2}
\sup_{\alpha\geq 0} F_{\m}(\alpha X)=F_{\m}(\hat\alpha X)=F((\hat\alpha X)\wedge 1).
\end{equation}

Now, thanks to Assumption \eqref{eq:EX>0} we obtain 
\begin{align} 
\SR_{\m}(X) = \sup_{Y\in L^0_+} \SR(X-Y)\geq \sup_{K\in\R}\SR(X\wedge K)>0.\label{eq:SRm1}
\end{align}
Next, by direct computation for $\E[Z]\geq 0$ and by Jensen's inequality for $\E[Z]<0$ one obtains for any $Z\in L^2$  
\begin{equation}\label{eq:FtoSR}
\max_{\alpha\geq 0}F(\alpha Z) = 1-(1+(\SR(Z)\vee 0)^2)^{-1} = g(\SR(Z)\vee 0).
\end{equation}
Observe that $g:z\mapsto 1-(1+z^2)^{-1}$ is strictly increasing on $\R_+$ with a continuous strictly increasing inverse function $g^{-1}:[0,1) \to \R_+$
\begin{equation}\label{ginv}
g^{-1}(y)=\sqrt{(1-y)^{-1}-1}.
\end{equation}
Therefore the left-hand side of \eqref{eq:FtoSR} uniquely identifies $\SR(Z)$ if $\E[Z]\geq 0$.

Apply this observation to \eqref{eq:SRm1} to obtain
\begin{align}
\SR_{\m}(X)&= \sup_{Y\in L^0_+}g^{-1}\left(\max_{\alpha\geq 0}F(\alpha(X-Y))\right) 
= g^{-1}\left(\sup_{\alpha\geq 0}\sup_{Y\in L^0_+}F(\alpha(X-Y))\right)\notag\\
&= g^{-1}\left(\sup_{\alpha\geq 0}\sup_{Y\in L^0_+}F(\alpha X-Y)\right)= g^{-1}\left(\max_{\alpha\geq 0}F_{\m}(\alpha X)\right),\label{eq:SRm3}
\end{align}
where the last equality follows from \eqref{eq:181213.1} and \eqref{eq:um2}. Observe that \eqref{eq:FmFOC} implies
$$ \E\left[(\hat\alpha X)\wedge 1\right] = \E\left[\left((\hat\alpha X)\wedge 1\right)^2\right],$$ 
which in turn gives 
$$ F((\hat\alpha X)\wedge 1) = g(\SR((\hat\alpha X)\wedge 1)).$$
This, \eqref{eq:um2}, and \eqref{eq:SRm3} proves the first equality in \eqref{eq:SRm2}. The second equality follows from homogeneity of the Sharpe ratio and the last from the inequality \eqref{eq:SRm1}.
\end{proof} 
\begin{remark}
Identity \eqref{eq:SRm2} shows that the monotone Sharpe ratio $\SR_{\m}$ is equal to the `arbitrage-adjusted Sharpe ratio' of \citet{cerny.03} for investment opportunities with positive mean and non-zero downside in $L^2$. We also remark that the Sharpe ratio bound in the good-deal pricing methodology of \citet{cochrane.saa-requejo.00} is in reality an upper bound on the monotone Sharpe ratio $\SR_{\m}$.
\end{remark}
\section{Optimal MMV investment and free cash-flow streams}
Denote the optimal strategies from Theorem~\ref{T:1} by $\hat{\vartheta}^x\in\scr{A}$ and $\hat{\vartheta}^x_{\m}\in\scr{A}_{\m}$, respectively. Using the relations (\ref{eq: CD}) we now
study the optimal portfolio allocation for monotone mean--variance preferences
$$u_{\MMV}(0)=\sup_{\vartheta \in \scr{T}} F_{\MMV}(\vartheta \cdot S_{T}).$$
Observe that due to Assumption~\ref{A:M} and Theorem~\ref{T:1} one has
$$ 0\leq u(0)\leq u_{\m}(0)<\frac{1}{2}.$$
\begin{theorem}\label{T:2}
Under Assumptions~\ref{A:S} and \ref{A:M} one has
\begin{equation} \label{eq:uMMV}
 u_{\MMV}(0)=\max_{\vartheta \in \scr{A}_{\m}} F_{\MMV}(\vartheta \cdot S_{T}) 
= \left((1-2u_{\m}(0))^{-1}-1  \right)/2,
\end{equation}
The optimal monotone mean--variance trading strategy in \eqref{eq:uMMV} equals 
\begin{equation}\label{eq:thetaMMV}
\hat{\vartheta}_{\MMV}^0=(1-2u_{\m}(0))^{-1}\hat{\vartheta}_{\m}^0.
\end{equation}
\end{theorem}
\begin{proof}
From (\ref{eq: CD}) and Theorem~\ref{T:1} we obtain
\begin{equation}\label{eq:ucm}
u_{\MMV}(0)=\sup_{c_{\m}\in \R}\sup_{\vartheta \in \scr{A}_{\m}}\E[U_{\m}(\vartheta \cdot S_{T}-c_{\m})+c_{\m}]
=\sup_{c_{\m}\in \R}u_{\m}(-c_{\m})+c_{\m}.
\end{equation}
Due to the self-similarity of $U_{\m}$ and the cone property of $\scr{A}_{\m}$ we have, just as in \citet{cerny.al.12},
\begin{align}
\hat{\vartheta}^x_{\m} &= (1-x)^+\hat{\vartheta}^0_{\m},\label{eq:thetam}\\
u_{\m}(x) &= 1/2 + ((1-x)^+)^2(u_{\m}(0)-1/2), \label{eq:um}
\end{align}
where \eqref{eq:um} follows by substituting \eqref{eq:thetam} into \eqref{eq:Fm}. 

Now substitute \eqref{eq:um} into \eqref{eq:ucm} 
and optimize over $c_{\m}$ to obtain 
\begin{align}\label{eq:cm}
\hat{c}_{\m}&=(1-2u_{\m}(0))^{-1}-1
\end{align}
together with \eqref{eq:uMMV}. Formula \eqref{eq:thetaMMV} now follows from \eqref{eq:thetam} with $x=-\hat{c}_{\m}$.
\end{proof}
\begin{remark}
Making use of the explicit formula for the truncated quadratic utility in \eqref{eq:Um} the first-order condition for the optimization over $c_{\m}$ in \eqref{eq:ucm} reads
$$ \E[(\vartheta\cdot S_T-\hat{c}_{\m}-1)^-]=1.$$
At the same time, \eqref{eq:Fcm} implies that 
$$ \hat\vartheta^0_{\MMV} = \arg\max_{\vartheta\in\mathcal{A}_{\m}}F_{\MV}\left((\vartheta\cdot S_T)\wedge (1+\hat{c}_{\m})\right).$$
This provides an alternative characterization of the optimal strategy, obtained previously  in \citet[Theorem~4.1]{maccheroni.al.09} for a one-period model.
\end{remark}
Mirroring the proof of Theorem~\ref{T:2} with standard MV preferences one obtains an analogous link between $u(0)$ and $u_{\MV}(0)$, 
\begin{equation}\label{eq:uMV}
u_{\MV}(0)=\sup_{\vartheta \in \scr{T}} F_{\MV}(\vartheta \cdot S_{T})=\left((1-2u(0))^{-1}-1  \right)/2.
\end{equation}
\begin{definition}
We say that one can extract a free cash-flow stream if there is a tame strategy $\vartheta\in\mathcal{T}$ and a non-negative random variable $Z\in L^{0}$ with $P(Z>0)>0$ such that 
$$F_{\MV}(\vartheta \cdot S_{T}-Z)\ge u_{\MV}(0).$$
\end{definition}
We are now in a position to formulate the main result on the availability of free cash-flow streams. 
To this end recall the function $g^{-1}$ in \eqref{ginv}.
\begin{theorem}\label{T:3} Under Assumptions~\ref{A:S} and \ref{A:M} the following statements hold.
\begin{enumerate}[label=(\arabic*),ref=(\arabic*)]
\item\label{item:1} The highest Sharpe ratio attainable by a tame zero-cost strategy is arbitrarily close to and does not exceed $\sqrt{2u_{\MV}(0)}$. This Sharpe ratio is attained but not exceeded in class $\scr{A}$ by the strategy $\hat\vartheta^0$,
\begin{equation} \label{eq:SR4}
\sup_{\vartheta\in\mathcal{T}}\SR(\vartheta\cdot S_T) 
= \max_{\vartheta\in\mathcal{A}}\SR(\vartheta\cdot S_T)
=g^{-1}\left(\max_{\vartheta\in\mathcal{A}}F(\vartheta\cdot S_T)\right)
=\SR(\hat\vartheta^0\cdot S_T).
\end{equation}
\item\label{item:2} The highest Sharpe ratio attainable by a tame zero-cost strategy after extracting a non-negative cash flow is arbitrarily close to and does not exceed $\sqrt{2u_{\MMV}(0)}$. This Sharpe ratio is attained but not exceeded in class $\scr{A}_{\m}$ by the strategy $\hat\vartheta_{\m}^0$,
\begin{equation}\label{eq:SRm4}
\sup_{\vartheta\in\mathcal{T}}\SR_{\m}(\vartheta\cdot S_T) 
= \max_{\vartheta\in\mathcal{A}_{\m}}\SR_{\m}(\vartheta\cdot S_T)
=g^{-1}\left(\max_{\vartheta\in\mathcal{A}_{\m}}F_{\m}(\vartheta\cdot S_T)\right)
=\SR_{\m}(\hat\vartheta_{\m}^0\cdot S_T).
\end{equation}
\item Provided $u_{\m}>0$ the maximal monotone Sharpe ratio in item~\ref{item:2} satisfies 
$$ \SR_{\m}(\hat\vartheta_{\m}^0\cdot S_T) = \SR((\hat{\vartheta}_{\m}^0\cdot S_T)\wedge 1).$$ 
The corresponding free terminal cash flow equals $$(1-\hat{\vartheta}_{\m}^0\cdot S_T)^+.$$
\item\label{item:4} By definition $u_{\MV}(0)\leq u_{\MMV}(0)$. The following are equivalent:
	\begin{enumerate}[label=(\alph*),ref=(\alph*)]
		\item\label{item:a} $u_{\MV}(0) = u_{\MMV}(0)$;
		\item\label{item:b} $u(0)=u_{\m}(0)$;
		\item\label{item:c} $\hat\vartheta^0\cdot S_T \leq 1$;
		\item\label{item:d} the variance-optimal $\sigma$--martingale measure is not signed.
	\end{enumerate}
\item\label{item:5} If there is $\overline\vartheta\in\mathcal{A}_{\m}$ and $0\neq Y\in L^0_+$ such that 
$\SR(\overline\vartheta\cdot S_T -Y) = \SR(\hat\vartheta^0\cdot S_T)$ then $\SR_{\m}(\hat\vartheta^0_{\m}\cdot S_T) > \SR(\hat\vartheta^0\cdot S_T)$.
\end{enumerate}
\end{theorem}
\begin{proof} Due to Jensen's inequality one can discard strategies with negative mean wealth. The proofs for zero mean wealth strategies are trivial hence we detail only the case where there exist strategies with positive (possibly infinite) mean and restrict attention only to those without further mention.

(1--2) Exploit the identities \eqref{eq:FtoSR} and \eqref{eq:SRm3} with $X=\vartheta\cdot S_T$. Observe that $\scr{T}$, $\scr{A}$, 
and $\scr{A}_{\m}$ are cones so multiplication by $\alpha > 0$ maps these sets onto themselves. This yields \eqref{eq:SR4} 
and \eqref{eq:SRm4}. In \eqref{eq:SRm4} apply \eqref{eq:uMMV} to obtain $\SR_{\m}(\hat\vartheta^0_{\m}\cdot S_T)=g^{-1}(u_{\m}(0))=\sqrt{2u_{\MMV}(0)}$. Proceed analogously in \eqref{eq:SR4} to obtain $\SR(\hat\vartheta^0\cdot S_T)=\sqrt{2u_{\MV}(0)}$.

(3) The condition $u_{\m}>0$ excludes the case $E[\vartheta^0_{\m}\cdot S_T]\leq 0$ therefore we can apply \eqref{eq:SRm2} with $X=\vartheta^0_{\m}\cdot S_T$.  The cone property of $\mathcal{A}_{\m}$, the optimality of $\vartheta^0_{\m}$, and the uniqueness of $\hat\alpha$ yield $\hat\alpha =1$ in \eqref{eq:SRm2} and the statement follows.

(4) Identities \eqref{eq:uMMV} and \eqref{eq:uMV} show equivalence between \ref{item:a} and \ref{item:b}. Due to the inclusion 
$\mathcal{M}_2^a\subseteq \mathcal{M}_2^s$ and the uniqueness of the variance-optimal measure, the duality in Theorem~\ref{T:1} 
implies $u(0)=u_{\m}(0)$ if and only if the dual optimizer in~\eqref{ux}, that is the variance-optimal measure, is in $\mathcal{M}^a_2$. 
This proves the equivalence between \ref{item:b} and \ref{item:d}. The equivalence between \ref{item:c} and \ref{item:d} follows from equation (3.16) and Proposition~3.13 in \citet{cerny.kallsen.07}.

(5) Argue by contradiction, assuming $u(0)=u_{\m}(0)$. Because $\SR(\hat\vartheta^0\cdot S_T)\geq 0$ we obtain $\E[\bar\vartheta\cdot S_T]\in(0,\infty]$. Now \eqref{eq:SRm2} with $X= \bar\vartheta\cdot S_T$ and hypothesis yield $\hat\alpha>0$ such that
$$ \SR_{\m}(X)=\SR(X-Y)=\SR(\hat\alpha(X-Y))=\SR((\hat\alpha X)\wedge 1)=g^{-1}\left(F_{\m}(\hat\alpha X)\right).$$
From here and \eqref{eq:SRm4} we conclude $\hat\vartheta_{\m}^0
=\hat\alpha \bar\vartheta\in \arg\max_{\vartheta\in\mathcal{A}_{\m}} F_{\m}(\vartheta\cdot S_T)$ has the property
\begin{equation}\label{over bliss}
\P(\hat\vartheta_{\m}^0\cdot S_T>1)>0.
\end{equation}
Because $u(0)=u_{\m}(0)$ the dual optimizers for $F$ and $F_{\m}$ must coincide by the same argument as in  item~\ref{item:4}. Now, Fenchel inequality implies that 
\begin{equation}\label{fenchel}
\hat\vartheta_{\m}^0\cdot S_T = \hat\vartheta^0\cdot S_T \text{ on the event } (\hat\vartheta^0\cdot S_T < 1),
\end{equation}
see also Theorem~4.10 (a)(iii) in \citet{biagini.cerny.11}. Equality of value functions also implies by item \ref{item:4}
$\hat\vartheta^0\cdot S_T \leq 1$. This, \eqref{over bliss}, and \eqref{fenchel} yield
\begin{equation}\label{arbitrage}
0\neq \hat\vartheta_{\m}^0\cdot S_T - \hat\vartheta^0\cdot S_T \in L^0_+.
\end{equation} 
Denote by $\Qu$ the equivalent measure from Assumption~\ref{A:M}. Admissibility in Definition~\ref{D:1} requires
$$\E^{\Qu}[\hat\vartheta^0\cdot S_T]=0\geq \E^{\Qu}[\hat\vartheta^0_{\m}\cdot S_T]$$
which together with $\Qu\sim\P$ contradicts the inequality in \eqref{arbitrage}.
\end{proof}
The previous theorem shows that one cannot extract a free cash-flow stream (FCFS) in the market $\mathcal{A}_{\m}$ \emph{if and only if} one cannot extract an FCFS from the mean--variance efficient portfolio $\hat\vartheta^0\in\mathcal{A}$. This is not entirely obvious in advance because first $\mathcal{A}_{\m}$ is a strict superset of $\mathcal{A}$ in general, and second in principle there could have been MV inefficient allocations in $\mathcal{A}$ that might have become very MV efficient after an FCFS extraction. 

\citet[Section~4]{cui.al.12} and \citet{trybula.zawisza.19} observe that specific diffusion models do not allow a free cash-flow stream. The next corollary identifies two very generic situations where an FCFS is not available, cf. also \citet[Theorem~3.3]{bauerle.grether.15}.
\begin{corollary}
Assume~\ref{A:S} and \ref{A:M}. If $\mathcal{M}^s_2$ is a singleton or if the price process $S$ is continuous then the extraction of a non-zero free cash-flow stream inevitably leads to a lower maximal Sharpe ratio over all left-over wealth distributions.
\end{corollary}
\begin{proof}
Under both hypotheses the variance-optimal measure is in $\mathcal{M}^e_2$; in the first case it follows by assumption and in the second it is the consequence of Theorem~1.3 in \citet{delbaen.schachermayer.96b}. Hence by item~\ref{item:4} of Theorem~\ref{T:3} 
$u_{\MV}(0)=u_{\MMV}(0)$ and by item \ref{item:5} any FCFS extraction must lead to a strictly lower maximal Sharpe ratio. 
\end{proof}

\end{document}